\documentclass[preprint,11pt,3p]{elsarticle}

\usepackage[numbers]{natbib}

\usepackage{amssymb}
\usepackage{graphicx}

\usepackage{pgf,tikz}

\usetikzlibrary{snakes}
\tikzstyle{printersafe}=[snake=snake,segment amplitude=0 pt]

\usetikzlibrary{arrows}
\usetikzlibrary{shapes}
\usepackage{amsmath}
\usepackage{caption}

\usepackage{todonotes}
\usepackage[linesnumbered,ruled]{algorithm2e}
\usepackage{pgf, tikz}
\usetikzlibrary{graphs, quotes, graphs.standard}
\usetikzlibrary{matrix}
\usepackage{tabularx}
\usepackage{arydshln,leftidx,mathtools}
\usepackage{setspace}

\usepackage{comment}
\usepackage{mathtools}

{\endarray\hspace*{-0.5\arraycolsep}\endBmatrix}
{\endarray\endbmatrix}

\usepackage{subfiles}
\usepackage{enumerate}
\usepackage{bm}
\usepackage[T1]{fontenc}
\usepackage[utf8]{inputenc}
\usepackage[english]{babel}
\usepackage{float}

\usepackage{accents}

\newtheorem{theorem}{Theorem}
\newtheorem{proposition}{Proposition}

\newtheorem{definition}{Definition}

\newtheorem{corollary}{Corollary}
\newtheorem{remark}{Remark}

\newenvironment{proof}{ {\bf Proof:}} 

\begin{document}
\onehalfspacing

\begin{frontmatter}

\title{Facets of the Total Matching Polytope for bipartite graphs}

\author{Luca Ferrarini}
\ead{l.ferrarini3@campus.unimib.it}

\address{Università di Pavia, Dipartimento di Matematica ``F. Casorati''}


\begin{abstract}
The {\it Total Matching Polytope} generalizes the Stable Set Polytope and the Matching Polytope.
In this paper, we give the perfect formulation for Trees and we derive two new families of valid inequalities, the {\it balanced biclique inequalities} which are always facet-defining and the {\it non-balanced lifted biclique inequalities} obtained by a lifting procedure, which are facet-defining for bipartite graphs.
Finally, we give a complete description for Complete Bipartite Graphs.

\end{abstract}


\end{frontmatter}

\section{Introduction}
Let $G=(V,E)$ be a simple, loopless and undirect graph, and let $D = V \cup E$ be the set of its elements.
We say that the elements $a,b \in D$ are {\it adjacent} if $a$ and $b$ are adjacent vertices, or if they are incident edges, or if $a$ is an edge incident to a vertex $b$.
If two elements $a,b \in D$ are not adjacent, they are {\it independent}.
A \textit{stable set} is an independent set of vertices, instead a \textit{matching} is an independent set of edges.
A \textit{total matching} is a subset $T \subseteq D$ where the elements are pairwise independent.
A subset $C \subseteq V \cup E$ is a \textit{total cover} of $G$ if it covers
all the elements of $G$.
The Total Matching Problem asks for a total matching of maximum size.
This problem generalizes both the Stable Set Problem, where we look for a stable set of maximum size and the Matching Problem, where instead we look for a matching of maximum size.
In particular, a matching is called {\it perfect} if it covers all the vertices, that is, has size $\frac{1}{2}|V|$. 
We define $\nu_T(G) := \max \{ |T| : T \mbox{ is a total matching}\}$, $\nu(G):= \max \{ |M| : M \mbox{ is a matching}\}$ and $\alpha(G):= \max \{|S| : S \mbox{ is a stable set}\}$.
%
The first work on the Total Matching Problem appeared in \cite{NordHaus}, where the authors derive lower and upper bounds on the size of a maximum total matching.
In particular, the authors show that:
\begin{center}
    $\nu_{T}(G) \geq \max\{\alpha(G),\nu(G)\}$
\end{center}
In \cite{TotalMatching}, they find a relation between $\nu_{T}(G)$ and $\tau(G):=\min \{|C| : C \mbox{ is a total cover}\}$, indeed
they show that:
\begin{center}
    $\tau(G) \leq \nu_{T}(G)$
\end{center}
In \cite{Manlove}, D.F. Manlove provides a concise survey of the algorithmic complexities of the decision problems related to the previous parameters.
The author reports that $\nu_{T}(G)$ can be computed in polynomial time for Trees and it is NP-complete for bipartite, chordal and arbitrary graphs.
From a polyhedral point of view, many studies of packing polytopes associated to the Stable Set Problem and the Matching Problem have been proposed.
The Stable Set Polytope is the convex hull of the incidence vectors of all stable sets and the Matching Polytope is the convex hull of the incidence vectors of all matchings.
Many valid and facet-defining inequalities have been investigated for Stable Set Polytope, see \cite{Letchford,Padberg1973,Chvatal,Oriolo,quasi-line,Rossi2001,Ventura,Rebennack2011,Yuri}.
In particular, complete linear descriptions have been obtained for classes of graphs as line-graphs and quasi-line graphs, see \cite{quasi-line,Edmonds1965}.
The authors in \cite{Polyhedra} propose the first polyhedral approach for the Total Matching Problem deriving facet-defining inequalities for the polytope associated with it. 
The convex hull of the incidence vectors of total matchings is called the Total Matching Polytope and it is denoted by $P_{T}(G)$.
Given a total matching $T$, the corresponding characteristic vector is defined as follows:
\begin{equation*}
  \chi[T]=\left\{
  \begin{array}{@{}ll@{}}
    z_{a}=1 & \text{if}\ a \in T \subseteq D  \\
    z_{a}=0 & \text{otherwise}.
  \end{array}\right.
\end{equation*} 
where $z = (x,y) \in \{0,1\}^{|V|} \times \{0,1\}^{|E|}$, $x$ correspond to the vertex variables and $y$ to the edge variables.
Hence, The {\bf Total Matching Polytope} of a graph $G=(V,E)$ is defined as:
\[
P_{T}(G) := \mbox{conv}\{\chi[T] \subseteq \mathbb{R}^{|V|+|E|} \mid T \subseteq D \mbox{ is a total matching} \}.
\]
%
%
Despite the strong connection with the Matching Problem, the Total Matching Problem is less studied in the operations research literature.
In particular, significant results are obtained only for structured graphs, as cycles, paths, full binary trees, hypercubes, and complete graphs, \cite{Leidner2012}. 
Motivated by the study of the Stable Set Polytope and the Matching Polytope for bipartite graphs, we expect to have a nice linear complete description of the Total Matching Polytope for bipartite graphs.
For this reason, in this paper we mainly focus on the facial structure of $P_{T}(G)$ for bipartite graphs.
We present two new families of facet-defining inequalities for the Total Matching Polytope and derive complete characterizations of Trees and Complete Bipartite Graphs.
%

\paragraph{Our contributions} The main results of this paper are:
\begin{itemize}
    
    \item Complete description of Trees and Complete Bipartite Graphs.
    
    \item New families of facet-defining inequalities, the {\it balanced biclique inequalities} and {\it non-balanced lifted biclique inequalities}.

\end{itemize}

\paragraph{Outline} The outline of this paper is as follows.
In the next paragraph, we fix the notation.
In Section 2, we study the Total Matching Polytope for Trees. 
In Section 3, we present new families of valid inequalities, the {\it balanced biclique inequalities}, which are facet-defining for all the graphs, and the {\it non-balanced biclique inequalities}, which are facet-defining for bipartite graphs.
Finally, in Section 4, we conclude as future works.
\paragraph{Notation}

Given a graph $G=(V,E)$, we define $n = |V|$ and $m = |E|$. 
For a vertex $v \in V$, we denote by $\delta(v)$ the set of edges incident to $v$ and by $N_{G}(v)$ the set of vertices adjacent to $v$. 
The degree of a vertex is $|\delta(v)|$, in particular, we denote by $\Delta(G):= \max\{|\delta(v)| \mid v \in V\}$. 
For a subset of vertices $U \subseteq V$, let $G[U]$ be the subgraph induced by $U$ on $G$.
We define $\delta(U):=\{e\in E \mid e=\{u,v\}, u \in U, v \in V \setminus U \}$. A total matching of maximum cardinality is denoted by $\nu_T(G)$.
%
%
%
%
A graph is \textit{chordal} if every cycle of length greater or equal than four has a \textit{chord}, that is, there is an edge connecting two non consecutive vertices of the cycle.
Given a graph $G$, a balanced biclique $K_{r,r}$ is a complete bipartite graph 
whose the cardinality of the two partitions of vertices is the same.

\section{Complete description of Trees}

First, we recall the basic important inequalities of $P_{T}(G)$.
\begin{proposition}
$P_{T}(G)$ has the following valid inequalities:
\begin{align}
\label{m6:c1} & x_v + \sum_{e \in \delta(v)} y_{e} \leq 1 & \forall v \in V \\
\label{m6:c2} & x_{v} + x_{w} + y_{e} \leq 1  & \forall e=\{v,w\} \in E \\
\label{m6:c3} & x_{v},y_{e} \geq 0 & \forall v \in V, \forall e \in E .
\end{align}
\end{proposition}
In \cite{Polyhedra} the authors prove that these inequalities are facet-defining.  
Since $\nu_{T}(G)$ is polynomial for Trees, we study the linear description of such graphs.
The following definition will be used to prove the main result of this Section.
\begin{definition}
Given a graph $G$, the total graph $T(G)$ of $G$ is a graph with
vertex set the vertices and edges of $G$, and two vertices are adjacent in $T$ 
if and only if their corresponding elements are either adjacent or incident in $G$.
\end{definition}
%
\begin{theorem}
Let $G$ be a Tree. Then $P_{T}(G)$ is completely defined by inequalities
\eqref{m6:c1} -- \eqref{m6:c3}.
\end{theorem}

\begin{proof}
Consider the total graph $T(G)$ of $G$. 
In \cite{Yannakakis}, it is proved that a graph is a Tree if and only if its total graph is chordal. 
Since chordal graphs are perfect graphs, and stable sets of the total graph correspond to total matchings of $G$, see \cite{Polyhedra}, we have that
$STAB(T(G))=\{x \in \mathbb{R}^{|V(T(G))|} \mid \sum\limits_{v \in K}x_v \leq 1, $ for every clique $K$ of $T(G), x_v \geq 0, \forall v \in V(T(G)) \}$.
Finally, maximal cliques of the Stable Set Polytope of $T(G)$ correspond to basic inequalities \eqref{m6:c1} -- \eqref{m6:c3} of the initial graph.
This completes the proof.
\qed
\end{proof}
This permits us to give an alternative polyhedral proof of the optimization problem for Trees.
\begin{corollary}
The optimization problem on $P_{T}(G)$ for a Tree graph can be solved in polynomial time.
\end{corollary}
%
%
\section{New families of facet-defining inequalities}
The basic facet-defining inequality \eqref{m6:c2} can be seen as a balanced biclique $K_{1,1}$.
Hence, we derive a generalization of this inequality in the following.
In \cite{Leidner2012}, the author proves that $\nu_{T}(K_{r,r})=r$.
Then, we have a natural upper bound to obtain the following valid inequality.
%
\begin{theorem}
Let $G$ be a graph and $K_{r,r}$ be an induced balanced biclique of $G$. Then, the balanced biclique inequality:
\begin{equation}\label{complete}
 \sum\limits_{v \in V(K_{r,r})}{x_{v}}+\sum\limits_{e \in E(K_{r,r})}{y_{e}} \leq r
\end{equation}
is facet-defining for $P_{T}(G)$.
\end{theorem}
\begin{proof}
Let $V(K_{r,r}):=A_1 \cup A_2$, where $A_1:=\{v_{1}, \dots, v_{r} \}$ and $A_2:=\{v_{r+1}, \dots, v_{2r}\}$.
Let $F$ be a face of $P_{T}(G)$, $F = \{z \in P_{T}(G) \mid \lambda^{T}z = \lambda_{0} \}$ and let $\Tilde{F} = \{z \in P_{T}(G) \mid \pi^{T}z = \pi_{0} \}$ be the face corresponding to the inequality \eqref{complete}, such that $\Tilde{F} \subseteq F$. 
We want to prove that there exists $a \in \mathbb{R}$ such that $(\lambda,\lambda_0) = a(\pi,\pi_0)$.
%
%
Since $K_{r,r}$ is a $r$-regular bipartite graph, then $E(K_{r,r})$ can be partitioned in $r$ perfect matchings $M_1, \dots, M_r$. 
%
First, we observe that $\chi[M_{i}] \in \Tilde{F}$, $\forall i = 1, \dots, r$.
Consider one of the perfect matchings, let $M_{i}$ and an edge $e=\{u,v\} \in M_{i}$.
Now, define the total matchings $T_{u} := (M_{i} \setminus \{e\}) \cup \{u\}$ and $T_{v} := (M_{i} \setminus \{e\}) \cup \{v\}$.
Note that, since the cardinality of the total matchings described is equal to the cardinality of the perfect matchings, this implies that $\chi[T_{u}],\chi[T_{v}] \in \Tilde{F}$.
We have that $\lambda^{T} \chi[M_{i}] = \lambda^{T} \chi[T_{u}] = \lambda^{T} \chi[T_{v}]$ and so $\lambda_{u}=\lambda_{v}=\lambda_{e}$.
By construction, every vertex is touched by all the perfect matchings, so $\lambda_v = \lambda_{e}$, $\forall e \in \delta(v)$, and
%
thus $\lambda_{i}=a \in \mathbb{R}$ for all the coefficients on $G[K_{r,r}]$. 
By fixing one of the total matchings $T$ introduced, we obtain also that $\lambda^{T}\chi[T]= \lambda_0 = a\pi_0$.
Now, fix a perfect matching $M$ on $G[K_{r,r}]$ and consider a vertex $w \notin V(K_{r,r})$.
Since $M \cap \{w\} = \emptyset$, $T_{w}:=M \cup \{w\}$ is a total matching, whose characteristic vector lies on $\Tilde{F}$.
This implies that $\lambda_w = 0$, $\forall w \notin V(K_{r,r})$.
Then, consider $T_{A_i}^{e}:=A_i \cup \{e\}$, where $e \notin \delta(A_i)$, for $i=1,2$.
It turns out that $T_{A_i}^{e}$ is a total matching and in particular its characteristic vector lies on $ \Tilde{F}$.
Thus $\lambda_e = 0$, $\forall e \notin E(K_{r,r})$.
This completes the proof, since we have proved that $(\lambda,\lambda_0) = a(\pi,\pi_0)$.
\qed
\end{proof}
Now, we focus on the computational complexity of the separation problem associated with balanced biclique inequalities.
%
%
In \cite{MEB}, the authors prove that it is NP-hard to compute a vertex maximum biclique.
The problem studied in \cite{MEB} is a specific instance of the Maximum Weighted Total Biclique Problem (MWTBP), which calls for a balanced biclique of maximum weight on vertices and edges.
Thus, we derive that in general, it is NP-hard to find the most violated balanced biclique inequality.
%
Since balanced biclique inequalities are facet-defining for the Total Matching Polytope, it is natural to ask if a non-balanced biclique generates
a facet-defining inequality.
Now, consider a general non-balanced biclique
$K_{r,s}$, with $s>r$. 
The biclique-inequality corresponding to the graph $K_{r,s}$ reads as follows:
\begin{align}
    \sum\limits_{v \in V(K_{r,s})}x_v + \sum\limits_{e \in E(K_{r,s})}y_e \leq s. 
\end{align}
It turns out that these inequalities are valid in general, but not facet-defining.
To get a facet-defining inequality we modify the coefficients.
By applying a sequential lifting, \cite{Zemel,Padberg1973}, we start the process with a suitable subset of variables.
From now on, let $K_{r,s}$ be a non-balanced biclique such that $V(K_{r,s})= R \cup S$ where $R:=\{v_1,v_2, \dots, v_r\}$ and $S:=\{w_1,w_2, \dots, w_s\}$ with  $s>r$ be the partitions of the vertices.

\begin{proposition}\label{lifting biclique}
Let $K_{r,s}$ with $s>r>1$ be a non-balanced biclique.
Then, the \textit{non-balanced lifted biclique inequalities}:
\begin{align}\label{biclique}
    \sum\limits_{i \in R}\alpha_i^{t}x_{i}+ \sum\limits_{j \in S}x_{j}+ \sum\limits_{e \in E(K_{r,s})} y_e \leq s & & \forall t \in R
\end{align}
where:
\begin{equation*}
  \alpha_{i}^{t}=\left\{
  \begin{array}{@{}ll@{}}
    s-(r-1) & \text{if}\ t = v_i, \\
    1 & \text{otherwise}.
  \end{array}\right.
\end{equation*} 
are facet-defining for $P_T(K_{r,s})$.
\end{proposition}

\begin{proof}
Let $N := \{0,1\}^{s+r+(s\times r)} \cap P_{T}(K_{r,s})$.
We want to lift the valid inequality $\sum\limits_{v \in S}x_v + \sum\limits_{e \in E(K_{r,s})}y_e \leq s$ 
of $N':=\mbox{conv}(N) \cap \{x_v = 0, v \in R\}$ 
into a facet-defining $\sum\limits_{i \in R}\alpha_{i}x_{i} + \sum\limits_{j \in S}x_{j} + \sum\limits_{e \in E(K_{r,s})}y_e \leq s$ of $\mbox{conv}(N)=P_{T}(K_{r,s})$.
We perform a sequential lifting of the coefficients according to the ordering of the index set $\{1,2, \dots, r\}$ of vertices of $R$. 
Now, consider the largest coefficient relative to $x_{v_1}$:
\begin{align*}
  \alpha_{v_1} := & s -\max\sum\limits_{i \in S}x_{i}+ \sum\limits_{e \in E(K_{r,s})}y_e \\  
  \mbox{s.t.} \quad & x_{v_{1}}=1, z \in N, \\
                    & x_{j} = 0, j \in R\setminus\{v_1\}.
\end{align*}
The optimum value of the maximization problem is $r-1$, since fixing the vertex $v_1$ we can exclude all the vertices of the other side, and since $K_{r,r}$ is an induced subgraph, the cardinality of a maximum total matching can be achieved by a matching $M$ of size $r-1$. 
%
Now, we claim that $\alpha_i = 1$ for $i \in \{v_2,v_3, \dots, v_r\}$.
%
%
\begin{align*}
    \alpha_{v_2} := & s - \max\sum\limits_{j \in S}x_{j}+ \sum\limits_{e \in E(K_{r,s})}y_e + (s-r+1)x_{v_1} \\ 
    \mbox{s.t.} \quad & x_{v_2}=1, z \in N, \\ 
                      & x_{j} = 0, j \in R\setminus\{v_1, v_2\}.
\end{align*}
Now, fixing the vertex $v_2$, it is easy to see that the optimal value is achieved by setting $x^{*}_{v_1}=1$ and $\chi[M_2]$, where $M_2$ is a matching of size $r-2$ induced by the vertices $v_3,\dots, v_r$ of a balanced biclique of size $r-2$. Thus, the optimal value is $(s-1)$.
We obtain that $\alpha_{v_2} = 1$.
At the step $i$ of the sequence we have: 
    
%
\begin{align*}
    \alpha_{v_i} := & s - \max \sum\limits_{j \in S}x_{w_j}+ \sum\limits_{e \in E(K_{r,s})}y_e + (s-r+1)x_{v_1}+\sum\limits_{t=2}^{i-1}x_{v_t} \\
    \mbox{s.t.} \quad  & x_{v_i}=1, z \in N, \\
                       & x_{j} = 0, j \in R\setminus\{v_{1}, \dots, v_{i-1}\}. 
\end{align*}
Repeating the same reasoning we obtain that $\alpha_{v_i}=1$.
We conclude that $(s-r+1)x_{v_1} + \sum\limits_{i=2}^{r}x_{v_i} + \sum\limits_{j=1}^{s}x_{w_j} + \sum\limits_{e \in E(K_{r,s})}y_e \leq s$ is a valid inequality for $P_{T}(K_{r,s})$.
Notice that the computation of the coefficients depends on the choice of the ordering where only the first coefficient of the sequence gets the maximum value different from one. 
So, consider $\sigma:\{1,2, \dots, r\} \longrightarrow \{1,2, \dots, r\}$ such that $ \sigma(i)=i+1$.
Choosing all the sequences with respect to the permutation $\sigma$, we get all the required facet-defining inequalities.
Now, we prove that after the lifting process the final inequality is facet-defining for $P_{T}(K_{r,s})$.
%
%
Let $F$ be the face induced by a non-balanced lifted inequality, w.l.o.g. suppose that the coefficient of the first vertex is $\alpha_{v_{1}}=s-r+1$ and let $F'$ be a face induced by $\pi^{T}z \leq \pi_{0}$, suppose that $F \subseteq F'$.
Consider a vertex $w_j \in S$ and an edge $e$ incident to $w_j$.
Now, it turns out that $\chi[S \setminus \{w_j\} \cup \{e\}] \in F$,
thus $\pi_{w_j} = \pi_e, \forall e \in \delta(w_j)$.
Then, fix a vertex $v_i \in R, i \neq 1$ and an edge $ f \in \delta(v_i)$.
We have that $\chi[R \setminus \{v_i\} \cup \{f\}] \in F$, 
this implies that $\pi_{v_i} = \pi_f, \forall f \in \delta(v_i)$, with $v_i \neq v_1$.
Since every edge has exactly one end-point in $R$ and the other in $S$, we have that $\pi_{v} = \pi_{e} = \pi_{w}, \forall e = \{v,w\},v \neq v_1$.
%
%
%
%
%
%
%
%
%
Until now, we obtain that $\pi_v = \pi_w = \pi_e, \ \forall v,w \in V(K_{r,s}) \setminus \{v_{1}\}, \forall e \in E(K_{r,s})$.
For the coefficient $\pi_{v_{1}}$, $\chi[R]= \chi[S]$ implies that $\pi_{v_{1}} + \sum\limits_{i=2}^{r}\pi_{v_{i}} = \sum\limits_{j=1}^{s}\pi_{w_{j}}$, thus $\pi_{v_{1}} = (s-r+1)a$, for a scalar $a \in \mathbb{R}$.
%
%
This completes the proof.
\qed
\end{proof}
%
%
\begin{figure}[!htbp]\label{biclique K_{2,3}}
\center
  \begin{tikzpicture}[scale = 0.90 ]
    \coordinate (a1) at (0,1);
    \coordinate (a2) at (0,3);
    \coordinate (b1) at (2,4);
    \coordinate (b2) at (2,2);
    \coordinate (b3) at (2,0);

    \draw[ line width = 0.05 cm]
    (a1)--(b1);
    \draw[ line width =0.05cm]
    (a1)--(b2);
    \draw[ line width = 0.05 cm] (a1)--(b3);
    \draw[ line width = 0.05 cm] (a2)--(b1);
    \draw[ line width = 0.05 cm] (a2)--(b2);
    \draw[ line width = 0.05 cm] (a2)--(b3);
   
    \draw[fill=red] (a1) circle [radius = 0.15 cm];
    \draw[fill=red] (a2) circle [radius = 0.15 cm];
    \draw[fill=blue] (b1) circle [radius = 0.15 cm];
    \draw[fill=blue] (b2) circle [radius = 0.15 cm];
    \draw[fill=blue] (b3) circle [radius = 0.15 cm];
    
    \node at (-0.5,1) {$v_2$};
    \node at (-0.5,3) {$v_1$};
    \node at (2.5,0) {$v_5$};
    \node at (2.5,2) {$v_4$};
    \node at (2.5,4) {$v_3$};

  \end{tikzpicture} 
  \caption{biclique $K_{2,3}$}
 \end{figure}
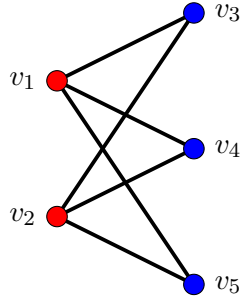
\begin{remark}
Observe that there is no lifting procedure for the coefficients of the edges of $K_{r,s}$.
\end{remark}
\begin{proof}
Let $G = K_{r,s}$.
Suppose by contradiction that there is an edge $e \in K_{r,s}$ such that $\beta_e >1$ for a non-balanced lifted biclique.
Hence, the corresponding inequality reads as follows $\sum\limits_{v \in V(K_{r,s})}x_v + \sum\limits_{f \in E(K_{r,s})\setminus\{e\}}y_f+\beta_ey_{e} \leq s$.
Let $V_S:= V(K_{r,s}) \setminus V(G[K_{r,r}])$ where $K_{r,r}$ is an induced balanced biclique of $G$
and fix a perfect matching $M$ on $G[K_{r,r}]$.
Then, $\chi[M\setminus\{e\}] + \beta_e + \chi[V_S] > s$ violates the inequality.
\qed
\end{proof}
Now, we can see an easy direct application for the biclique $K_{2,3}$.
The corresponding inequalities read as follow:
\begin{equation*}
    2x_{v_1}+x_{v_2}+x_{v_3}+x_{v_4}+x_{v_5}+\sum\limits_{e \in E(K_{2,3})}y_e \leq 3, 
\end{equation*}
\begin{equation*}
    x_{v_1}+2x_{v_2}+x_{v_3}+x_{v_4}+x_{v_5}+\sum\limits_{e \in E(K_{2,3})}y_e \leq 3.
\end{equation*}
The following proposition shows that the lifting procedure exposed in Proposition (\ref{lifting biclique}) is exhaustive and maximal, that is, it generates all the possible facet-defining induced biclique inequalities.
\begin{proposition}
Consider a non-balanced biclique $K_{r,s}$. Then, the inequality:
\begin{equation}\label{allcoeff}
    \sum\limits_{i=1}^{|R|} \pi_{v_i}x_{v_{i}} + \sum\limits_{j=1}^{|S|}x_{w_j} + \sum\limits_{e \in E(K_{r,s})}y_e \leq s,
\end{equation}
such that $\sum\limits_{i=1}^{|R|}\pi_{v_i}=s-r+1$ with at least two coefficients different from one is not facet-defining for $P_{T}(K_{r,s})$.
\end{proposition}

\begin{proof}
We show that the face $F'$ induced by the inequality (\ref{allcoeff}) is properly contained in the face $F$ induced by a non-balanced lifted biclique inequality.
%
%
Suppose that w.l.o.g. the first $l < r$ vertices have the corresponding coefficients different from one, thus define $U_R := \{v_{1},v_{2}, \dots, v_{l}\}$
and $L:= R \setminus U_R$.
First, notice that the inequality $(s-r+1)x_{v_1} + \sum\limits_{i=2}^{r} x_{v_{i}} + \sum\limits_{j=1}^{s}x_{w_j} + \sum\limits_{e \in E(K_{r,s})}y_e \leq s $ is tight for all total matchings satisfying \eqref{allcoeff} at equality.
In particular, every characteristic vector of a total matching belonging to $F$ is of the form $\chi[U_R] + \chi[T_L]$, where $T_L$ is a total matching composed by a stable set on the vertices of $L$ and a matching in $K_{r,s}[L]$ induced on the remaining vertices.
It is easy to construct such a total matching $T_L$, for example, we can take $T_L$ as the set $L$.
%
Notice also that $\chi[T_L \cup U_R] \in F$.
Then, we exhibit a point $x \in F$ such that $x \in F \setminus F'$ and $F ' \subseteq F$.
Consider the total matching $T_{1}$ defined as:
\begin{equation*}
  \chi[T_1]:=\left\{
  \begin{array}{@{}ll@{}}
    x_{v_1} = 1 &   \\
    y_{e}=1 & \forall e=\{v_i,w_{i-1}\} \ \forall i=2,3, \dots, r \\
    0 & \text{otherwise}.
  \end{array}\right.
\end{equation*} 
So, there is at least one more solution in $F$.
This concludes the proof.
\qed
\end{proof}

\begin{theorem}
Let $G$ be a bipartite graph. Then, the non-balanced lifted biclique inequalities \eqref{biclique}
are facet-defining for $P_{T}(G)$.
\end{theorem}

\begin{proof}
Let $V(G):= A_1 \cup A_2$ be the partition of the vertices of $G$ and $V(K_{r,s}) := B_1 \cup B_2$ a non-balanced biclique of $G$.
We denote by $F$ the face induced by a non-balanced lifted biclique inequality.
%
%
By Proposition (\ref{lifting biclique}), we have $|V(K_{r,s})|+|E(K_{r,s})|$ affinely independent points belonging to $F$.
Now, the vectors of the form $\chi[B_i] + \chi[\{v\}]$, where $v \in A_i \setminus B_i$ for $i=1,2$, lie on $F$ and they are linearly independent.
Then, consider $\chi[B_i] + \chi[\{e\}] $, where $e \notin \delta(B_i)$, for $i=1,2$.
Note that they are characteristic vectors of total matchings and they are linearly independent.
The final matrix, having as columns these vectors, has the following form:
\begin{center}
$\left[
\begin{array}{c|c|c}
A_{K_{r,s}} & B_{K_{r,s}} & C_{K_{r,s}} \\ \hline
\mathbf{0} & \widetilde{I_{v}} & \mathbf{0} \\ \hline
\mathbf{0} & \mathbf{0} & \widetilde{I_{e}}\\
\end{array}\right],
$
\end{center}
where $A_{K_{r,s}},B_{K_{r,s}},C_{K_{r,s}}$ represent the matrices corresponding to the characteristic vectors of total matchings on $K_{r,s}$,
and $\widetilde{I_{v}},\widetilde{I_{e}}$ are the identity matrices relative to the elements not in $K_{r,s}$.
Since the matrix has maximum rank, we can conclude. 
\qed
\end{proof}
%
Now, we are ready to state the final theorem, since we have the full list of facet-defining inequalities describing $P_{T}(K_{r,s})$.
\begin{theorem}
$P_{T}(K_{r,s})$ is completely defined by:
\begin{itemize}
    \item Basic inequalities \eqref{m6:c1} -- \eqref{m6:c3},
    \item Balanced biclique inequalities \eqref{complete},
    \item Non-balanced lifted biclique inequalities \eqref{biclique}.
\end{itemize}
\end{theorem}
\section{Conclusion and future works}
In this paper, we have introduced two new families of facet-defining inequalities for the Total Matching Polytope
and we have found complete linear description for Trees and Complete Bipartite graphs.
As future work, we plan to give a complete linear description of the Total Matching Polytope for bipartite graphs.
%

\section*{Acknowledgments}
The research was partially supported by the Italian Ministry of Education, University and Research (MIUR): Dipartimenti di Eccellenza Program (2018--2022) - Dept. of Mathematics ``F. Casorati'', University of Pavia.

I am deeply indebted to Stefano Gualandi for discussions and insightful observation on the topic.

\bibliographystyle{apalike}
\bibliography{biblio}

\end{document}